\documentclass
[%preprint,
3p,times,final,authoryear]{elsarticle}
\usepackage{amsthm}

\usepackage{natbib}

\usepackage{dsfont}
\usepackage{graphicx}
\usepackage{mathrsfs}
\usepackage{amssymb}
\usepackage{amsbsy}
\usepackage{amsmath}
\usepackage{amsfonts}
\usepackage{latexsym}
\usepackage[latin1]{inputenc}
\usepackage{txfonts}
\usepackage{color}

\usepackage{hyperref}

\newcommand{\abs}[1]{\left\vert#1\right\vert}

\newcommand{\pq}[1]{\left[#1\right]}
\newcommand{\pt}[1]{\left(#1\right)}

\newcommand{\di}[1]{\mathrm{d}#1}

\newcommand{\Data}{X^{(n)}} %X_1, \ldots, X_n}
\newcommand{\DataXe}{X^{(n)}} %{X_1,\, \ldots,\, X_n}
\newcommand{\dataXe}{x^{(n)}} %{x_1,\, \ldots,\, x_n}
\newcommand{\pa}[1]{\lfloor#1\rfloor}

\def\1{\mathbf{1}}

\newtheorem{theorem}{\bf Theorem}
\newtheorem{lemma}{\bf Lemma}

\newtheorem{proposition}{\bf Proposition}
\newtheorem{definition}{\bf Definition}

\newtheorem{Rem}{\bf Remark}
\newtheorem{Ex}{\bf Example}

\journal{}

\makeatletter
\renewcommand\@biblabel[1]{}
\renewenvironment{thebibliography}[1]
     {\section*{\refname}%
      \list{}%
           {\leftmargin0pt
            \@openbib@code
            \usecounter{enumiv}
            }%
      \sloppy
      \clubpenalty4000
      \@clubpenalty \clubpenalty
      \widowpenalty4000%
      \sfcode`\.\@m}
     {\def\@noitemerr
       {\@latex@warning{Empty `thebibliography' environment}}%
      \endlist}
\makeatother

\begin{document}

\begin{frontmatter}



\title{Sharp sup-norm Bayesian curve estimation
%\\Sharp uniform-norm posterior concentration
}
\author{Catia Scricciolo\corref{cor1}}

%\author{Catia Scricciolo\corref{cor1}}

%\fntext[]{}
\cortext[cor1]{Corresponding author.
%. Tel: +39 02 5836.5684; Fax: +3902 5836.5634
} \ead{catia.scricciolo@univr.it}
%\address{Address\fnref{label3}}

%%\ead[url]{}
%\fntext[]{}
%\cortext[cor1]{}
%\address{Address\fnref{label3}}


%% \address[label1]{<Tel. +39 02 58365684>}
%%\address[label2]{<address>}

\address{Department of Economics, University of Verona, Via Cantarane 24, 37129 Verona, Italy}

\begin{abstract}

Sup-norm curve estimation is a fundamental statistical problem and, in principle, a premise
for the construction of confidence bands for infinite-dimensional
parameters. In a Bayesian framework, the issue of whether the
sup-norm-concentration-of-posterior-measure approach proposed by \citet{GN11}, which involves
solving a testing problem exploiting concentration properties of kernel and projection-type density estimators around their expectations,
can yield minimax-optimal rates is herein settled in the affirmative beyond conjugate-prior settings
%by applying McDiarmind's inequality, which outperforms Talagrand's inequality for the $L^1$-distance between
%a kernel-type estimator and its expectation, thus
obtaining sharp rates for common prior-model pairs like random histograms, Dirichlet Gaussian or Laplace mixtures,
which can be employed for density, regression or quantile estimation.
% and for estimating linear functionals of the distribution function like quantiles.

%, meanwhile
%providing some \textcolor[rgb]{0.98,0.00,0.00}{gap filling} contributions
%to the topic and
%pointing out some open problems.
\end{abstract}


\end{frontmatter}

% \linenumbers


\section{Introduction}
The study of the frequentist asymptotic behaviour of Bayesian nonparametric (BNP) procedures has initially focused on
%contraction rates in
the Hellinger or $L^1$-distance loss, see \citet{SW2001} and \citet{GGvdvV2000}, but an extension and generalization of the results to
$L^r$-distance losses, $1\leq r\leq \infty$, has been the object of two recent contributions by \citet{GN11} and \citet{CAS2014}.
Sup-norm estimation has particularly attracted attention as it may constitute
the premise for the construction of confidence bands whose geometric structure can be easily visualized
and interpreted. Furthermore, as shown in the example of Section \ref{ss:quantile}, the study of sup-norm posterior contraction
rates for density estimation can be motivated as being an intermediate step for the final assessment of convergence rates for
quantile estimation.

While the contribution of \cite{CAS2014} has a more prior-model specific flavour,
the article by \citet{GN11} aims at a unified understanding of the drivers of the
asymptotic behaviour of BNP procedures by
developing a new approach to the involved testing problem
constructing nonparametric tests that have good exponential bounds
on the type-one and type-two error probabilities that rely on
concentration properties of kernel and projection-type density estimators around their expectations.

Even if \citet{GN11}'s approach can only be useful if a fine control of
the approximation properties of the prior support is possible,
it has the merit of replacing the entropy condition for sieve sets with approximating
conditions. However, the result, as presented in their Theorem 2 (Theorem 3),
can only retrieve minimax-optimal rates for $L^r$-losses when
$1\leq r\leq 2$, while rates deteriorate by a genuine power of $n$, in fact $n^{1/2}$, for $r>2$.
Thus, the open question remains whether
their approach can give the right rates for $2< r \leq \infty$ for non-conjugate priors and
sub-optimal rates are possibly only an artifact of the proof.
We herein settle this issue in the affirmative by refining their result and proof and
showing in concrete examples that this approach retrieves the right rates.

The paper is organized as follows. In Section \ref{s:mainresult}, we state the main result whose proof is postponed to
\ref{App0}. Examples concerning different statistical settings like
density and quantile estimation are presented in Section \ref{s:examples}.

\section{Main result}\label{s:mainresult}
In this section, we describe the set-up and present the main contribution of this note.
Let $\bigl((\mathcal X,\, \mathcal A,\,P),\, P\in \mathcal P\bigr)$ be a collection of probability measures
on a measurable space $(\mathcal X,\,\mathcal A)$ that
possess densities with respect to some $\sigma$-finite dominating measure $\mu$.
Let $\Pi_n$ be a sequence of priors on $(\mathcal P,\,\mathcal B)$, where $\mathcal B$
is a $\sigma$-field on $\mathcal P$ for which the maps $x\mapsto p(x)$ are jointly measurable relative to
$\mathcal A \otimes \mathcal B$. Let $X_1,\,\ldots,\,X_n$ be i.i.d. (independent, identically distributed)
observations from a common law $P_0\in \mathcal P$ with density $p_0$ on $\mathcal X$ with respect to $\mu$, $p_0=\di P_0/\di \mu$.
For a probability measure $P$ on $(\mathcal{X},\,\mathcal{A})$ and
an $\mathcal{A}$-measurable function $f:\,\mathcal{X}\rightarrow \mathbb{R}^k$, $k\geq1$, let
$P f$ denote the integral $\int f\,\di P$, where, unless otherwise specified,
the set of integration is understood to be the whole domain.
When this notation is applied to the empirical measure $\mathbb{P}_n$ associated with
a sample $X^{(n)}:=(X_1,\,\ldots,\,X_n)$, namely the discrete uniform measure on the sample values,
this yields $\mathbb{P}_n f=n^{-1}\sum_{i=1}^nf(X_i)$. For each $n\in\mathbb{N}$, let
$\hat p_n(j)(\cdot)={n}^{-1}\sum_{i=1}^nK_j(\cdot,\,X_i)$ be a kernel or projection-type density estimator based
on $X_1,\,\ldots,\,X_n$ at resolution level $j$, with $K_j$ as in
Definition \eqref{def:approxseq} below. Its expectation is then equal to
$P_0^n\hat p_n(j)(\cdot)=P_0K_j(\cdot,\,X_1)=K_j(p_0)(\cdot)$,
where we have used the notation $K_j(p_0)(\cdot)=\int K_j(\cdot,\,y)p_0(y)\,\di y$.
In order to refine \citet{GN11}'s result, we use
concentration properties of $\|\hat p_n(j)-K_j(p_0)\|_1$ around its expectation by applying
McDiarmind's inequality for bounded differences functions.

The following definition, which corresponds to Condition 5.1.1 in \citet{GN15}, is essential for the main result.
\begin{definition}\label{def:approxseq}
\emph{Let $\mathcal X=\mathbb{R}$, $\mathcal X=[0,\,1]$ or $\mathcal X=(0,\,1]$. The sequence of operators
\[K_j(x,\,y):=2^j K(2^jx,\, 2^jy), \quad x,\,y\in\mathcal X, \quad j\geq0,\]
is called \emph{an admissible approximating sequence} if it satisfies one of the
following conditions:
\begin{itemize}
\item[a)] convolution kernel case, $\mathcal X = \mathbb{R}$: $K(x,\,y)=K(x-y)$,
where $K\in L^1(\mathbb{R})\cap L^\infty(\mathbb{R})$,
integrates to $1$ and is of bounded $p$-variation for some finite $p\geq1$ and right (left)-continuous;\\[-0.6cm]
\item[b)] multi-resolution projection case, $\mathcal X = \mathbb{R}$: $K(x,\,y)=\sum_{k\in\mathbb{Z}}\phi(x-k)\phi(y-k)$,
with $K_j$ as above or $K_j(x,\,y)=K(x,\,y)+\sum_{\ell=0}^{j-1}\sum_k\psi_{lk}(x)\psi_{lk}(y)$, where
$\phi,\,\psi\in L^1(\mathbb{R})\cap L^\infty(\mathbb{R})$ define an $S$-regular wavelet basis, have bounded $p$-variation for some
$p\geq1$ and are uniformly continuous, or define the Haar basis, see Chapter 4, \emph{ibidem};\\[-0.6cm]
\item[c)] multi-resolution case, $\mathcal X = [0,\,1]$: $K_{j,bc}(x,\,y)$ is the projection kernel at resolution $j$
of a Cohen-Daubechies-Vial (CDV) wavelet basis, see Chapter 4, \emph{ibidem};\\[-0.6cm]
\item[d)] multi-resolution case, $\mathcal X = (0,\,1]$: $K_{j,per}(x,\,y)$ is the projection kernel at resolution $j$
of the periodization of a scaling function satisfying b), see (4.126) and (4.127), \emph{ibidem}.
\end{itemize}}
\end{definition}
\begin{Rem}
\emph{A useful property of $S$-regular wavelet bases is the following: there exists a non-negative measurable
function $\Phi\in L^1(\mathbb{R})\cap L^\infty(\mathbb{R})$ such that $|K(x,\,y)|\leq \Phi(|x-y|)$
for all $x,\,y\in\mathbb{R}$, that is, $K$ is dominated by a bounded and integrable convolution kernel
$\Phi$.}
\end{Rem}
In order to state the main result, we recall that a sequence of positive real numbers $L_n$
is \emph{slowly varying} at $\infty$ if, for each $\lambda >0$, it holds that
$\lim_{n\rightarrow\infty}(L_{[\lambda n]}/L_n)=1$. Also, for $s\geq0$, let $L^1(\mu_s)$ be the space of
$\mu_s$-integrable functions, $\di \mu_s(x):=(1+|x|)^s\di x$,
equipped with the norm $\|f\|_{L^1(\mu_s)}:=\int |f(x)|(1+|x|)^s\,\di x$.
\begin{theorem}\label{th:l_r-norms}
Let $\epsilon_n$ and $J_n$ be sequences of positive real numbers such that
$\epsilon_n\rightarrow0$, $n\epsilon_n^2\rightarrow\infty$ and $2^{J_n}=O(n\epsilon_n^2)$.
For each $r\in\{1,\,\infty\}$ and a slowly varying sequence
$L_{n,r}\rightarrow\infty$, let $\epsilon_{n,r}:=L_{n,r}\epsilon_n$. Suppose that, for $K$ as in Definition \eqref{def:approxseq},
with $K^2$, $\Phi^2$ and  $p_0$ that integrate $(1+|x|)^s$ for some $s>1$ in cases a) and b),
\begin{equation}\label{eq:approx}
\|K_{{J_n}}(p_0)-p_0\|_r=O(\epsilon_{n,r})
\end{equation}
and, for a constant $C>0$, sets
$\mathcal P_n\subseteq \{P\in\mathcal P:\, \|K_{{J_n}}(p)-p\|_r\leq C_{K}\epsilon_{n,r}\}$,
where $C_{K}>0$ only depends on $K$, we have
\begin{itemize}
\item[(i)] $\Pi_n(\mathcal P\setminus \mathcal P_n)\leq \exp{\bigl(-(C+4)n\epsilon_n^2\bigr)}$,
\item[(ii)] $\Pi_n\bigl(P\in\mathcal P:\,-P_0\log (p/p_0)\leq\epsilon_n^2,\,\,\, P_0\log^2 (p/p_0)\leq\epsilon_n^2\bigr)\geq
\exp{(-Cn\epsilon_n^2)}$.
\end{itemize}
Then, for sufficiently large $M_r>0$,
\begin{equation}\label{eq:postconc}
P_0^n\Pi_n\bigl(P\in\mathcal P:\,\|p-p_0\|_r\geq M_r\epsilon_{n,r}\mid \Data\bigr)\rightarrow0.
\end{equation}
If the convergence in \eqref{eq:postconc} holds for $r\in\{1,\,\infty\}$, then, for each $1< s < \infty$.
$P_0^n\Pi_n\bigl(P\in\mathcal P:\,\|p-p_0\|_s\geq M_s\bar\epsilon_n\mid \Data\bigr)\rightarrow0$,
where $\bar\epsilon_n:=(L_{n,1}\vee L_{n,\infty})\epsilon_n$.
\end{theorem}
The assertion, whose proof is reported in \ref{App0},
is an in-probability statement that the posterior mass outside a sup-norm ball
of radius a large multiple $M$ of $\epsilon_n$ is negligible.
The theorem provides the same sufficient conditions for deriving sup-norm posterior
contraction rates that are minimax-optimal, up to logarithmic factors, as in
\citet{GN11}. Condition (ii), which is mutuated from \citet{GGvdvV2000},
is the essential one: the prior concentration rate is
the only determinant of the posterior contraction rate at densities $p_0$
having sup-norm approximation error of the same order against a kernel-type
approximant, provided the prior support is almost the set of densities with the same approximation
property.

\section{Examples}\label{s:examples}
In this section, we apply Theorem \ref{th:l_r-norms}
to some prior-model pairs used for (conditional) density or regression estimation,
including random histograms, Dirichlet Gaussian or Laplace
mixtures, that have been selected in an attempt to reflect cases for which the issue of obtaining sup-norm posterior rates
was still open. We do not consider Gaussian priors or wavelets series
because these examples have been successfully worked out in
\citet{CAS2014} taking a different approach. We furthermore exhibit an example
with the aim of illustrating that obtaining sup-norm posterior contraction rates for density estimation can be motivated as
being an intermediate step for the final assessment of convergence rates for estimating single quantiles.

\subsection{Density estimation}
\begin{Ex}[Random dyadic histograms] \emph{For $J_n\in\mathbb{N}$,
consider a partition of $[0,\,1]$ into $2^{J_n}$ intervals (\emph{bins}) of equal length
$A_{1,2^{J_n}}=[0,\,2^{-J_n}]$ and $A_{j,2^{J_n}}=((j-1)2^{-J_n},\,j2^{-J_n}]$,
$j=2,\,\ldots,\,2^{J_n}$. Let $\mathrm{Dir}_{2^{J_n}}$ denote the Dirichlet distribution
on the $(2^{J_n}-1)$-dimensional unit simplex with all parameters equal to $1$.
Consider the random histogram %with $2^{J_n}$ bins
\[\sum_{j=1}^{2^{J_n}}w_{j,2^{J_n}}2^{J_n} 1_{A_{j,2^{J_n}}}(\cdot),\quad (w_{1,2^{J_n}},\,\ldots,\,w_{2^{J_n},2^{J_n}})\sim \mathrm{Dir}_{2^{J_n}}.\]
Denote by $\Pi_{2^{J_n}}$ the induced law on the space of probability measures with Lebesgue density on $[0,\,1]$.
Let $X_1,\,\ldots,\,X_n$ be i.i.d. observations from a density $p_0$ on $[0,\,1]$. Then, the Bayes' density estimator,
that is the posterior expected histogram, has expression
\[\hat p_n(x)=\sum_{j=1}^{2^{J_n}}\frac{1+N_{l(x)}}{2^{J_n}+n}2^{J_n} 1_{A_{j,2^{J_n}}}(x),\quad x\in [0,\,1],\]
where $l(x)$ identifies the bin containing $x$, i.e.,
$A_{l(x),2^{J_n}}\ni x$, and $N_{l(x)}$ stands for the number of observations falling into
$A_{l(x),2^{J_n}}$. Let $\mathcal C^\alpha([0,\,1])$ denote the class of Hölder continuous functions
on $[0,\,1]$ with exponent $\alpha>0$. Let
$\epsilon_{n,\alpha}:=\bigl(n/\log n\bigr)^{-\alpha/(2\alpha+1)}$ be the minimax
rate of convergence over $(\mathcal C^\alpha([0,\,1]),\,\|\cdot\|_\infty)$.}
\begin{proposition}
Let $X_1,\,\ldots,\,X_n$ be i.i.d. observations from a density $p_0\in\mathcal C^\alpha([0,\,1])$, with
$\alpha\in(0,\,1]$, satisfying $p_0>0$ on $[0,\,1]$. Let $J_n$ be such that
$2^{J_n}\sim \epsilon_{n,\alpha}^{1/\alpha}$. Then, for sufficiently large $M>0$,
$P_0^n\Pi_{2^{J_n}}(P:\,\|p-p_0\|_\infty\geq M \epsilon_{n,\alpha}\mid \DataXe)\rightarrow0$.
Consequently, $P_0^n\|\hat p_n-p_0\|_\infty\asymp \epsilon_{n,\alpha}$.
\end{proposition}
\emph{The first part of the assertion, which concerns posterior contraction rates,
immediately follows from Theorem \eqref{th:l_r-norms} combined with the
proof of Proposition 3 of \citet{GN11}, whose result, together with that of Theorem 3
in \citet{CAS2014}, is herein improved to the
minimax-optimal rate $\bigl(n/\log n\bigr)^{-\alpha/(2\alpha+1)}$ for every $0<\alpha\leq 1$.
The second part of the assertion, which concerns convergence rates for the histogram density estimator, is a
consequence of Jensen's inequality and convexity of $p\mapsto \|p-p_0\|_\infty$,
combined with the fact that the prior $\Pi_{2^{J_n}}$ is supported on densities uniformly bounded above by $2^{J_n}$ and
that the proof of Theorem \ref{th:l_r-norms} yields the exponential
order $\exp{(-Bn\epsilon_{n,\alpha}^2)}$ for the convergence of the posterior probability of the complement
of an $(M\epsilon_{n,\alpha})$-ball around $p_0$, in symbols,
$P_0^n\|\hat p_n-p_0\|_\infty< M\epsilon_{n,\alpha} + 2^{J_n} P_0^n\Pi_{2^{J_n}}(P:\,\|p-p_0\|_\infty
\geq M \epsilon_{n,\alpha}\mid \DataXe)\leq M \epsilon_{n,\alpha} + 2^{J_n} \exp{(- Bn\epsilon_{n,\alpha}^2)},$
whence $P_0^n\|\hat p_n-p_0\|_\infty=O(\epsilon_{n,\alpha})$.}
\end{Ex}

\smallskip

\begin{Ex}[Dirichlet-Laplace mixtures] \emph{Consider, as in \citet{S2011}, \citet{Gao2015}, a Laplace mixture prior $\Pi$
thus defined. For $\varphi(x):=\frac{1}{2}\exp{(-|x|)}$, $x\in\mathbb{R}$,
the density of a Laplace$\,(0,\,1)$ distribution, let
\begin{itemize}
\item $p_G(\cdot):=\int \varphi(\cdot-\theta)\,\di G(\theta)$ denote a mixture of Laplace densities with mixing distribution $G$,\\[-0.6cm]
\item $G\sim \mathrm{D}_\alpha$, the Dirichlet process with base measure $\alpha:=\alpha_{\mathbb{R}}\bar\alpha$,
for $0<\alpha_{\mathbb{R}}<\infty$ and $\bar \alpha$ a probability measure on $\mathbb{R}$.
\end{itemize}
\begin{proposition}\label{prop2}
Let $X_1,\,\ldots,\,X_n$ be i.i.d. observations from a density $p_{G_0}$, with $G_0$
supported on a compact interval $[-a,\,a]$. If $\alpha$ has support on $[-a,\,a]$
with continuous Lebesgue density bounded below away from $0$ and above from $\infty$,
then, for sufficiently large $M>0$, $P_0^n\Pi(P:\,\|p-p_0\|_\infty\geq M (n/\log n)^{-3/8}\mid \DataXe)\rightarrow0$.
Consequently, for the Bayes' estimator $\hat p_n(\cdot)=\int p_G(\cdot)\Pi(\di G\mid \Data)$
we have $P_0^n\|\hat p_n-p_0\|_\infty\asymp (n/\log n)^{-3/8}$.
\end{proposition}
\begin{proof}
It is known from Proposition 4 in \citet{Gao2015} that the small-ball probability estimate
in condition (ii) of Theorem \ref{th:l_r-norms}
is satisfied for $\epsilon_n=(n/\log n)^{-3/8}$. For the bias condition, we take
$\mathcal P_n$ to be the support of $\Pi$ and show that, for $2^{J_n}\sim \epsilon_n^{-1/3}=(n/\log n)^{1/8}$ and any symmetric density
$K$ with finite second moment, we have
$\|K_{{J_n}}(p_G)-p_G\|_\infty=O(\epsilon_n)$
uniformly over the support of $\Pi$. Indeed, by applying Lemma \ref{lemma1} with $\beta=2$,
for each $x\in\mathbb{R}$ it results $
|K_{{J_n}}(p_G)(x)-p_G(x)|^2  \leq
\|K_{{J_n}}(p_G)-p_G\|_2^2
\leq\int |\tilde \varphi(t)|^2|\tilde K(2^{-J_n}t)-1|^2\,\di t \sim
(2\pi)^{-1}(B_\varphi^2\times I_2[\tilde K])(2^{2J_n})^{-3}$,
which implies that both conditions \eqref{eq:approx} and (i) are satisfied.
The assertion on the Bayes' estimator follows from the same arguments laid out for random histograms
together with the fact that $p_G\leq 1/2$ uniformly in $G$.
\end{proof}
}
\end{Ex}

\begin{Ex}[Dirichlet-Gaussian mixtures] \emph{Consider, as in \citet{GvdV01, GvdV072b},
\citet{STG2013}, \citet{S2014}, a Gaussian mixture prior $\Pi\times G$
thus defined. For $\phi$ the standard normal density, let
\begin{itemize}
\item $p_{F,\sigma}(\cdot):=\int \phi_\sigma(\cdot-\theta)\,\di F(\theta)$ denote a mixture of Gaussian densities with mixing distribution $F$,\\[-0.6cm]
\item $F\sim \mathrm{D}_\alpha$, the Dirichlet process with base measure $\alpha:=\alpha_{\mathbb{R}}\bar\alpha$,
for $0<\alpha_{\mathbb{R}}<\infty$ and $\bar \alpha$ a probability measure on $\mathbb{R}$, which has
continuous and positive density $\alpha'(\theta)\propto e^{-b|\theta|^\delta}$ as $|\theta|\rightarrow\infty$,
for some constants $0<b<\infty$ and $0<\delta\leq 2$,\\[-0.6cm]
\item $\sigma\sim G$ which has continuous and positive
density $g$ on $(0,\,\infty)$ such that, for constants
$0<C_1,\,C_2,\,D_1,\,D_2<\infty$, $0\leq s,\,t<\infty$,
\[C_1\sigma^{-s}\exp{(-D_1\sigma^{-1}\log^t(1/\sigma)}\leq g(\sigma)\leq C_2\sigma^{-s}
\exp{(-D_2\sigma^{-1}\log^t(1/\sigma))}\]
for all $\sigma$ in a neighborhood of $0$.
\end{itemize}
Let $\mathcal C^\beta(\mathbb{R})$ denote the class of Hölder continuous functions
on $\mathbb{R}$ with exponent $\beta>0$.
Let
$\epsilon_{n,\beta}:=\bigl(n/\log n\bigr)^{-\beta/(2\beta+1)}$ be the minimax
rate of convergence over $(\mathcal C^\beta(\mathbb{R}),\,\|\cdot\|_\infty)$.
For any real $\beta>0$, let $\pa{\beta}$ stand for the largest integer strictly smaller than $\beta$.
\begin{proposition}\label{prop2}
Let $X_1,\,\ldots,\,X_n$ be i.i.d. observations from a density $p_0\in L^\infty(\mathbb{R})\cap C^\beta(\mathbb{R})$ such that condition $(ii)$
is satisfied for $\epsilon_{n,\beta}$. %If $\alpha$ has a
%continuous and positive density $\alpha'(\theta)\propto e^{-b|\theta|^\delta}$ as $|\theta|\rightarrow\infty$,
%for some constants $0<b<\infty$ and $0<\delta\leq 2$,
Then, for sufficiently large $M>0$, $P_0^n(\Pi\times G)((F,\,\sigma):\,\|p_{F,\sigma}-p_0\|_\infty\geq M \epsilon_{n,\beta}\mid \DataXe)\rightarrow0$.
\end{proposition}
\begin{proof}
Let $K\in L^1(\mathbb{R})$ be a convolution kernel such that
\begin{description}
\item[$\bullet$] $\int x^k K(x)\,\di x=\1_{\{0\}}(k)$, $k=0,\,\ldots,\, \pa{\beta}$, and $\int|x|^{\beta}|K(x)|\,\di x<\infty$,
\item[$\bullet$] the Fourier transform $\tilde K$ has $\textrm{supp}(\tilde K) \subseteq [-1,\,1]$.
\end{description}
Let $2^{J_n}\sim \epsilon_{n,\beta}^{1/\beta}$. For every $x\in \mathbb{R}$,
$|K_{{J_n}}(p_0)(x)-p_0(x)|\leq C_1(2^{- J_n})^\beta\lesssim \epsilon_{n,\beta}$, where the
constant $C_1\propto (1/\pa{\beta}!)\int|x|^{\beta}|K(x)|\,\di x$ does not depend on $x$.
Thus, $\|K_{{J_n}}(p_0)-p_0\|_\infty=O(\epsilon_{n,\beta})$.
For the bias condition, let $\underline{\sigma}_n:=E(n\epsilon_{n,\beta}^{2})^{-1}(\log n)^{\psi}$,
with $1/2<\psi<t$ and a suitable constant $0<E<\infty$.
For every $\sigma\geq\underline{\sigma}_n$ and uniformly in $F$,
%$(F,\sigma)\in\mathcal P_n:=\{(F,\sigma):\,F\in\mathscr{M}(\mathbb{R}),\,\}$,
%We show that $\mathcal P_n\subseteq\{P_{(F,\sigma)}:\,\|K_{J_n}(p_{F,\sigma})-p_{F,\sigma}\|_\infty\leq C_K\epsilon_{n,\infty}\}$ for every sufficiently large $n$.
\[\begin{split}
\|K_{J_n}(p_{F,\sigma})-p_{F,\sigma}\|_\infty
& =\sup_{x\in \mathbb{R}}\abs{\int \int K_{{J_n}}(u)[\phi_\sigma(x-v-u)-\phi_\sigma(x-v)]\,\di u\,\di F(v)}\\
& \leq\frac{1}{2\pi}\sup_{x\in \mathbb{R}}\int
\int |e^{-it(x-v)}||\tilde\phi_\sigma(t)||\tilde K_{}(2^{-J_n}t)-1|\,\di t\,\di F(v)\\
& \leq \frac{1}{\pi}
\int_{|t|>2^{J_n}} |\tilde\phi_\sigma(t)|\,\di t
\lesssim \underline{\sigma}_n^{-1}\exp{(-(\rho\underline{\sigma}_n2^{J_n})^2)}\lesssim n^{-1}<\varepsilon_{n,\beta}
\end{split}\]
because $(\underline{\sigma}_n2^{J_n})^2\propto(\log n)^{2\psi}\gtrsim(\log n)$ as $\psi>1/2$.  Now,
$G(\sigma<\underline{\sigma}_n)\lesssim \underline{\sigma}_n^{-s}\exp{(-[D_2\underline{\sigma}_n^{-1}\log^t(1/\underline{\sigma}_n)])}\lesssim \exp{(-(C+4)n\epsilon_n^2)}$ because $\psi<t$, which implies that the remaining mass condition $(ii)$ is satisfied.
\end{proof}
\begin{Rem}
\emph{Conditions on the density $p_0$ under which assumption (ii) of Theorem \ref{th:l_r-norms} is satisfied
can be found, for instance, in \citet{STG2013} and \citet{S2014}.}
\end{Rem}
}
\end{Ex}

\subsection{Quantile estimation}\label{ss:quantile}
For $\tau\in(0,\,1)$, consider the problem of estimating the $\tau$-quantile $q_{0}^\tau$
of the population distribution function $F_0$ from observations $X_1,\,\ldots,\,X_n$.
For any (possibly unbounded) interval $I\subseteq \mathbb{R}$ and function $g$ on $I$, define the Hölder norm as
\[
\|g\|_{\mathcal C^\alpha(I)}:=\sum_{k=0}^{\pa{\alpha}}\|g^{(k)}\|_{L^\infty(I)}+\sup_{x,\,y\in I:\,x\neq y}\frac{|g^{\pa{\alpha}}(x)-g^{\pa{\alpha}}(y)|}{|x-y|^{\alpha-\pa{\alpha}}}.
\]
Let $\mathcal C^0(I)$ denote the space of continuous and bounded functions on $I$ and
$\mathcal C^\alpha(I,\,R):=\{g\in \mathcal C^0(I):\,\|g\|_{\mathcal C^\alpha(I)}\leq R\}$, $R>0.$
%For any non-negative integer $\alpha$ and %any
%interval $I\subseteq \mathbb{R}$
%define the space $\mathcal C^\alpha(I)$ of all bounded continuous real-valued functions that are
%$\alpha$-times continuously differentiable on $I$, equipped with the norm
%\[\|f\|_{\alpha,\infty}=\sum_{k=0}^\alpha\|D^kf\|_\infty,\]
%with the convention that $D^0=id$ and $\mathcal C^0(I)=\mathcal C(I)$. For any non-integer $\alpha>0$, set
%\[\mathcal C^\alpha(I)=\pg{f\in \mathcal C^{[\alpha]}(I):\,\|f\|_{\alpha,\infty}:=\sum_{k=0}^{[\alpha]}\|D^kf\|_\infty+\sup_{x\neq y}
%\frac{|D^{[\alpha]}f(x)-D^{[\alpha]}f(y)|}{|x-y|^{\alpha-[\alpha]}}<\infty},\]
%where $[\alpha]$ denotes the integer part of $\alpha$. For any real $R>0$, define
%\[\mathcal C^\alpha(I,\,R)=\{f\in \mathcal C^\alpha(I):\,\|f\|_{\alpha,\infty}\leq R\}.
%\]
%Let $\mathcal Q(R)$ denote the set of all Lebesgue probability densities on $\mathbb{R}$ that are uniformly bounded by $R>0$.
%For $R,\,r,\,\zeta,\,U>0$ and the smoothness index $\alpha>0$ define the class
%\[\begin{split}
%\mathcal C^\alpha(R,\,r,\,\zeta,\,U)&:=\bigg\{f\in \mathcal Q(R):\, f \mbox{ has a $\tau$-quantile $q^\tau\in[-U,\,U]$ st}  \mbox{$f\in \mathcal C^\alpha([q^\tau-\zeta,\,q^\tau+\zeta],\,R)$ and $f(q^\tau)\geq r$}\bigg\}.
%\end{split}\]
\begin{proposition}\label{prop:inversion3}
Suppose that, given $\tau\in(0,\,1)$, there are constants $r,\,\zeta>0$ so that
$p_{0}(\cdot+q^\tau_{0})\in\mathcal C^\alpha([-\zeta,\,\zeta],\,R)$ and
\begin{equation}\label{eq:positivity}
\inf_{[q^\tau_{0}-\zeta,\,q^\tau_{0}+\zeta]} p_{0}(x)\geq r.
\end{equation}
Consider a prior $\Pi$ concentrated on probability measures having densities
$p(\cdot+q^\tau_{0})\in \mathcal C^\alpha([-\zeta,\,\zeta],\,R)$. If, for sufficiently large $M$,
the posterior probability
$
P_0^n\Pi(P:\,\|p-p_{0}\|_\infty\geq M \epsilon_{n,\alpha}\mid \Data)\rightarrow0
$, then, there exists $M'>0$ so that $P_0^n\Pi(|q^\tau- q^\tau_{0}|\geq M' \epsilon_{n,\alpha}^{1+1/\alpha}\mid \Data)\rightarrow0$.
\end{proposition}
\begin{proof}
We preliminarily make the following remark.
Let $F(x):=\int_{-\infty}^x p(y)\,\di y$, $x\in \mathbb{R}$. For $\tau\in(0,\,1)$,
let $q^\tau$ be the $\tau$-quantile of $F$. By Lagrange's theorem, there exists a %(random)
point $q^\tau_*$ between $q^\tau$
and $q^\tau_{0}$ so that $F(q^\tau)-F(q^\tau_{0})=p(q^\tau_*)(q^\tau- q^\tau_{0})$. Consequently,
\[
0=\tau-\tau =\int_{-\infty}^{q^\tau}p(x)\,\di x-\int_{-\infty}^{q^\tau_{0}}p_{0}(x)\,\di x=\int_{q^\tau_{0}}^{q^\tau}p(x)\,\di x+
\int_{-\infty}^{q^\tau_{0}}[p(x)-p_{0}(x)]\,\di x=p(q^\tau_*)(q^\tau- q^\tau_{0})+[F(q^\tau_{0})-F_{0}(q^\tau_{0})].
\]
If $p(q^\tau_*)>0$, % almost surely,
then
\begin{equation}\label{eq:translation}
q^\tau- q^\tau_{0}=-\frac{[F(q^\tau_{0})-F_{0}(q^\tau_{0})]}{p(q^\tau_*)}=
-\frac{[F(q^\tau_{0})-\tau]}{p(q^\tau_*)}.
\end{equation}
In order to upper bound $|q^\tau- q^\tau_{0}|$, by appealing to relationship \eqref{eq:translation},
we can separately control $|F(q_{0}^\tau)-F_{0}(q_{0}^\tau)|$ and
$p(q^\tau_\ast)$.
%Let $\psi_n:=n^{-(\alpha+1)/(2\alpha+1)}(\log n)^t$.
%For any real $\alpha>0$,
Let %$K\in L^2(\mathbb{R})$
the kernel function $K\in L^1(\mathbb{R})$ %serve per controllare $\|\hat{K}_b\|_2=\|K_b\|_2$
be such that %defined as the limit $\lim_{n\rightarrow\infty}\hat K_n$, where $\hat K_n(t)=\int_{-n}^ne^{-itx}K(x)\di x$, and the limit is taken in $L^2(\mathbb{R})$,
\begin{description}
\item[$\bullet$] $\int x^k K(x)\di x=1_{\{0\}}(k)$, $k=0,\,\ldots,\, \lfloor\alpha\rfloor+1$, and
$\int|x|^{\alpha+1}|K(x)|\,\di x<\infty$,\\[-0.5cm]
\item[$\bullet$] its Fourier transform $\tilde K$ has $\textrm{supp}(\tilde K) \subseteq [-1,\,1]$.
\end{description}
%, for $\pa{\alpha}$ the largest integer strictly smaller than $\alpha$,
By Lemma 5.2 in \citet{DRT2013},
\begin{equation}\label{eq:bias}
\sup_{p_{0}(\cdot+q_{0}^\tau)\in \mathcal C^{\alpha}([-\zeta,\,\zeta],\, R)}\abs{\int_{-\infty}^{q^\tau_{0}}[K_b\ast p_{0}-p_{0}](x)\,\di x}\leq Db^{\alpha+1},
\end{equation}
with $D:=[R/(\lfloor\alpha\rfloor+1)!+2\zeta^{-(\alpha+1)}]\int|x|^{\alpha+1}|K(x)|\,\di x$. Write
\[
F(q^\tau_{0})-F_{0}(q^\tau_{0})=
\int_{-\infty}^{q^\tau_{0}}[K_b\ast p_{0}-p_{0}](x)\,\di x +
\int_{-\infty}^{q^\tau_{0}}[K_b\ast(p-p_{0})](x)\,\di x +
\int_{-\infty}^{q^\tau_{0}}[p-K_b\ast p](x)\,\di x=:T_1+T_2+T_3.
\]
By inequality \eqref{eq:bias}, we have $|T_1|=O(b^{\alpha+1})$. By the same reasoning, $|T_3|=O(b^{\alpha+1})$.
We now consider $T_2$. Taking into account that $\int K(x)\,\di x=1$ and
\[
T_2:=[K_b\ast (F-F_{0})](q^\tau_{0})=\int\frac{1}{b}K\pt{\frac{q^\tau_{0}-u}{b}}(F-F_{0})(u)\,\di u
=-\int K(z)(F-F_{0})(q^\tau_{0}-bz)\,\di z
=\int K(z)(F_{0}-F)(q^\tau_{0}-bz)\,\di z,
\]
for some point $\xi$ between $q^\tau_{0}-bz$ and $q^\tau_{0}$ (clearly, $\xi$ depends on $q^\tau_{0},\,z,\,b$),
\[
\begin{split}
T_2=[K_b\ast (F-F_{0})](q^\tau_{0})\mp (F_{0}-F)(q^\tau_{0})
&=\int K(z)[(F_{0}-F)(q^\tau_{0}-bz)-(F_{0}-F)(q^\tau_{0})]\,\di z+ (F_{0}-F)(q^\tau_{0})\\
&=\int K(z)(-bz)[D^1(F_{0}-F)(\xi)]\,\di z + (F_{0}-F)(q^\tau_{0})\\
&=(-b)\int zK(z)[(p_{0}-p)(\xi)]\,\di z + (F_{0}-F)(q^\tau_{0}).
\end{split}
\]
%In virtue of relationship \eqref{eq:translation},
Then, $
F(q^\tau_{0})-F_{0}(q^\tau_{0})=T_1+T_3+ (-b)\int zK(z)[(p_{0}-p)(\xi)]\,\di z - [(F-F_{0})(q^\tau_{0})]$,
which implies that $2[F(q^\tau_{0})-F_{0}(q^\tau_{0})]=T_1+T_3+ (-b)\int zK(z)[(p_{0}-p)(\xi)]\,\di z$.
It follows that
$2|F(q^\tau_{0})-F_{0}(q^\tau_{0})|\leq |T_1|+|T_3|+ b\|p_{0}-p\|_\infty\int |z||K(z)|\,\di z$.
Taking into account that $\int |z||K(z)|\,\di z<\infty$,
$|T_1|=O(b^{\alpha+1})$ and
 $|T_3|=O(b^{\alpha+1})$, choosing $b=O(\epsilon_{n,\alpha}^{1/\alpha})$, we have
$|F(q^\tau_{0})-F_{0}(q^\tau_{0})|\lesssim |T_1|+|T_3|+ b\|p_{0}-p\|_\infty
\lesssim b^{\alpha+1}+b \|p_{0}-p\|_\infty\lesssim \epsilon_{n,\alpha}^{1+1/\alpha}$.
If $\|p-p_{0}\|_\infty\lesssim\epsilon_{n,\alpha}$ then, under condition
\eqref{eq:positivity}, $p(q^\tau_*)>r-\eta>0$ for every $0<\eta<r$.
In fact, for any interval $I\supseteq [q^\tau_{0}-\zeta,\,q^\tau_{0}+\zeta]$
that includes the point $q^\tau$ so that it also includes the intermediate point
$q^\tau_*$ between $q^\tau$ and $q^\tau_{0}$, for any $\eta>0$ we have
$
\eta\gtrsim \|p-p_{0}\|_\infty\geq \sup_I |p(x)-p_{0}(x)|
\geq  |p(\tilde x)-p_{0}(\tilde x)|$ for every $\tilde x \in I$.
It follows that
$p(q^\tau_*)>p_{0}(q^\tau_*)-\eta\geq\inf_{x\in[q^\tau_{0}-\zeta,\,q^\tau_{0}+\zeta]}p_{0}(x) -\eta\geq r-\eta$.
Conclude the proof by noting that, in virtue of \eqref{eq:translation},
$P_0^n\Pi(P:\,\|p-p_{0}\|_\infty < M\epsilon_{n,\alpha}\mid \Data)\leq
P_0^n\Pi(|q^\tau- q^\tau_{0}| < M' \epsilon_{n,\alpha}^{1+1/\alpha}\mid \Data)$.
The assertion then follows.
\end{proof}
\begin{Rem}\label{cor:L_3}
\emph{Proposition \ref{prop:inversion3} considers local
Hölder regularity of $p_0$, which seems natural for estimating single quantiles.
Clearly, requirements on $p_0$ are automatically satisfied if $p_0$ is globally Hölder regular and,
in this case, the minimax-optimal sup-norm rate is $\epsilon_{n,\alpha}=(n/\log n)^{-\alpha/(2\alpha+1)}$ so that
the rate for estimating single quantiles is
$\epsilon_{n,\alpha}^{1+1/\alpha}=(n/\log n)^{-(\alpha+1)/(2\alpha+1)}$.
The conditions on the random density $p$ are automatically satisfied if the prior is concentrated on probability measures
possessing globally Hölder regular densities.}
\end{Rem}

%%%%%%%%%%%%%%%%%%%%%%%%%%%%%%%%%%%%%%%%%%%%%%%%%%%%%%%%%%%%%%%%%%%%%%%%%%%%%%%%%%%%%%%%%%%%%%%%%%%%%%%%%%%%%%%%%%%%%%%%%%%%%%%%%%%%%%%%%

\appendix

\section{Proof of Theorem \ref{th:l_r-norms}}\label{App0}
\begin{proof}
Using the remaining mass condition (i) and the small-ball probability estimate (ii),
by the proof of Theorem 2.1 in \cite{GGvdvV2000}, it is enough to construct, for each $r\in\{1,\,\infty\}$,
a test $\Psi_{n,r}$ for the hypothesis
$$H_0:\,P=P_0\quad\mbox{\emph{vs.}}\quad H_1:\,\{P\in\mathcal{P}_n:\,\|p-p_0\|_r\geq M_r\epsilon_{n,r}\},$$
with $M_r>0$ large enough, where $\Psi_{n,r}\equiv\Psi_{n,r}(\DataXe;\,P_0):\,\mathcal X^n\rightarrow\{0,\,1\}$
is the indicator function of the rejection region of $H_0$, such that
\[\begin{split}
 P_0^n\Psi_{n,r}\rightarrow0 \,\,\, \mbox{ as $n\rightarrow\infty$} \,\,\,\,
 \mbox{ and }\,\,\,\, \sup_{P\in\mathcal{P}_n:\,\|p-p_0\|_r\geq M_r\epsilon_{n,r}}
P^n(1-\Psi_{n,r})\leq \exp{\bigl(-K_rM_r^2n\epsilon_{n,r}^2\bigr)} \, \mbox{ for sufficiently large $n$},
\end{split}\]
where $K_rM_r^2\geq (C+4)$, the constant $C>0$ being that appearing in (i) and (ii).
By assumption \eqref{eq:approx}, there exists a constant $C_{0,r}>0$ such that
$\|P_0^n\hat p_n-p_0\|_r=\|K_{J_n}(p_0)-p_0\|_r\leq C_{0,r}\epsilon_{n,r}$.
Define $T_{n,r}:=\|\hat p_n-p_0\|_r$. For a constant $M_{0,r}>C_{0,r}$, define the event
$A_{n,r}:=(T_{n,r}>M_{0,r}\epsilon_{n,r})$ and the test $\Psi_{n,r}:=1_{A_{n,r}}$.
For
\begin{itemize}
\item $r=1$, the triangular inequality $T_{n,1}\leq \|\hat p_n-P_0^n\hat p_n\|_1+\|P_0^n\hat p_n-p_0\|_1$ implies that, when
$T_{n,1}>M_{0,1}\epsilon_{n,1}$, $\|\hat p_n-P_0^n\hat p_n\|_1\geq T_{n,1} - \|P_0^n\hat p_n-p_0\|_1> M_{0,1}\epsilon_{n,1}-
\|P_0^n\hat p_n-p_0\|_1 \geq (M_{0,1}-C_{0,1})\epsilon_{n,1}$;
%because, by assumption, there exists a constant $c'$ such that $\|P_0^n\hat p_n-p_0\|_1=\|K_{J_n}(p_0)-p_0\|_1\leq c'\epsilon_{n,1}$.
\item $r=\infty$, we have $|\hat p_n(x)-p_0(x)|\leq |\hat p_n(x)-P_0^n\hat p_n(x)|+|P_0^n\hat p_n(x)-p_0(x)|
\leq \|\hat p_n-P_0^n\hat p_n\|_1 + \|P_0^n\hat p_n-p_0\|_\infty$ for every $x\in\mathbb{R}$.
It follows that $T_{n,\infty}\leq \|\hat p_n-P_0^n\hat p_n\|_1+\|P_0^n\hat p_n-p_0\|_\infty$,
which implies that, when $T_{n,\infty}>M_{0,\infty}\epsilon_{n,\infty}$,
$\|\hat p_n-P_0^n\hat p_n\|_1\geq T_{n,\infty} - \|P_0^n\hat p_n-p_0\|_\infty> M_{0,\infty}\epsilon_{n,\infty}-
\|P_0^n\hat p_n-p_0\|_\infty\geq (M_{0,\infty}-C_{0,\infty})\epsilon_{n,\infty}$.
%because by assumption there exists a constant $c'$ such that $\|P_0^n\hat p_n-p_0\|_\infty=\|K_{2^{-J_n}}\ast p_0-p_0\|_\infty<c'\epsilon_{n,\infty}$.
\end{itemize}
Let $h:\,\mathcal X^n\rightarrow[0,\,2]$
be the function defined as $h(\DataXe):=\|\hat p_n-P_0^n\hat p_n\|_1$.
Thus, for each $r\in\{1,\,\infty\}$, when $T_{n,r}>M_{0,r}\epsilon_{n,r}$,
the inequality $h(\DataXe)>(M_{0,r}-C_{0,r})\epsilon_{n,r}$ holds.
Therefore, to control the type-one error probability,
it is enough to bound above the probability on the right-hand side of the following display
\begin{equation}\label{eq:I_type_error}
P_0^n\Psi_{n,r}\leq P_0^n\bigl(h(\DataXe)>(M_{0,r}-C_{0,r})\epsilon_{n,r}\bigr),
\end{equation}
which can be done using McDiarmind's inequality, \cite{McDiarmid1989}.
Given any $\dataXe:=(x_1,\,\ldots,\,x_n)\in \mathcal X^n$, for each $1\leq i\leq n$,
let $x_i$ be the $i$th component of $\dataXe$ and $x_i':=(x_i+\delta)$ a perturbation of the
$i$th variable with $\delta\in \mathbb{R}$ so that $x_i'\in\mathcal X$.
Letting $e_i$ be the canonical vector with all zeros except for a $1$ in the $i$th position,
the vector with the perturbed $i$th variable can be expressed as $\dataXe+\delta e_i$.
If
\begin{itemize}
\item[$(a)$] the function $h$ has \emph{bounded differences}:
for some non-negative constants $c_1,\,\ldots,\,c_n$,
$$\sup_{\substack{\dataXe,\,x_i'}}
|h(\dataXe)-h(\dataXe+\delta e_i)|\leq c_i,\quad 1\leq i\leq n,$$
\item[$(b)$] $P_0^n h(\DataXe)=O(\epsilon_n)$,
\end{itemize}
then, for $C:=\sum_{i=1}^n c_i^2$, by McDiarmind's bounded differences inequality,
\[\forall\,t>0,\quad P_0^n\bigl(|h(\DataXe)- P_0^nh(\DataXe)|\geq t\bigr)\leq 2 \exp{\bigl(-2t^2/C\bigr)}.\]
We show that $(a)$ and $(b)$ are verified.
\begin{itemize}
\item[$(a)$] Using the inequality $||a|-|b||\leq|a-b|$, setting $\Phi=K$
under condition $a$) of Definition \eqref{def:approxseq},
\[\begin{split}
%\hspace*{-0.1cm}
\forall\,i\in\{1,\,\ldots,\,n\},\quad\sup_{\dataXe,\,x_i'}
|h(\dataXe)-h(\dataXe+\delta e_i)|
&=\sup_{\dataXe,\,x_i'}\int\Bigg[\Bigg|\frac{1}{n}\sum_{i=1}^nK_{J_n}(x,\,x_i)- K_{J_n}(p_0)(x) \Bigg|\\&
\hspace*{1.5cm}  - \Bigg|\frac{1}{n}\sum_{i\neq i'}^nK_{J_n}(x,\,x_i)
+ \frac{1}{n}K_{J_n}(x,\,x_i') - K_{J_n}(p_0)(x) \Bigg|\Bigg]\,\di x\Bigg|\\
& \leq \sup_{x_i,\,x_i'}\frac{1}{n}\|K_{J_n}(\cdot,\,x_i)-K_{J_n}(\cdot,\,x_i')\|_1
\leq \frac{2}{n}\|\Phi\|_1.
\end{split}\]
Hence, $h$ has bounded differences with $c_i=2\|\Phi\|_1/n$, $1\leq i\leq n$.
%In the convolution kernel case, $\sup_{\substack{\dataXe\\ x_i'\in \mathcal X}}
%|h(\dataXe)-h(\dataXe+\delta e_i)|\leq 2\|K\|_1/n$.
\item[$(b)$] By Theorem 5.1.5 in \citet{GN15},
$P_0^n h(\DataXe)\leq L\sqrt{2^{J_n}/n}=O(\epsilon_n)$, with the following upper bounds for the constant $L$:\\[-0.5cm]
\begin{description}
\item[$\bullet$] under conditions $a$) and $b$) of Definition \eqref{def:approxseq}, setting $\Phi=K$ in case $a$),
$L\leq \sqrt{2/(s-1)}\,\|\Phi^2\|_{L^1(\mu_s)}^{1/2}\|p_0\|_{L^1(\mu_s)}^{1/2}$;
\item[$\bullet$] under conditions $c$) and $d$), $L\leq C(\phi)(1\vee \|p_0\|_{1/2})^{1/2}$, where the constant $C(\phi)$ only depends on $\phi$.
\end{description}
\end{itemize}
For $\alpha\in(0,\,1)$, taking $t=\sqrt{2}\alpha(M_{0,r}-C_{0,r})\epsilon_{n,r}$,
$$P_0^n\bigl(|h(\DataXe)- P_0^nh(\DataXe)|\geq \sqrt{2}\alpha(M_{0,r}-C_{0,r})\epsilon_{n,r}\bigr)
\leq 2 \exp{\bigl(-\alpha^2(M_{0,r}-C_{0,r})^2n\epsilon_{n,r}^2/\|\Phi\|_1^2}\bigr).$$
By $(b)$, there exists a constant $L'\geq L$ so that
$P_0^n h(\DataXe)\leq L'\epsilon_n= (L'/L_{n,r}) \epsilon_{n,r}$. Hence,
$|h(\DataXe)- P_0^nh(\DataXe)| \geq h(\DataXe)-P_0^nh(\DataXe) \geq h(\DataXe)- (L'/L_{n,r}) \epsilon_{n,r}$.
Thus, for sufficiently large $L_{n,r}$ so that
$[(M_{0,r}-C_{0,r})-(L'/L_{n,r})]\geq \sqrt{2}\alpha(M_{0,r}-C_{0,r})$,
\[\begin{split}
P_0^n\bigl(\|\hat p_n-P_0^n\hat p_n\|_1\geq (M_{0,r}-C_{0,r})\epsilon_{n,r}\bigr)
%=P_0^n\Bigl(h(\DataXe)\geq (M_0-C)\epsilon_{n,r}\Bigr)
&\leq P_0^n\bigl(|h(\DataXe)- P_0^nh(\DataXe)|\geq \sqrt{2}\alpha(M_{0,r}-C_{0,r})\epsilon_{n,r}\bigr)\\
&\leq 2 \exp{\bigl(-\alpha^2(M_{0,r}-C_{0,r})^2n\epsilon_{n,r}^2/\|\Phi\|_1^2}\bigr).
\end{split}\]

We now provide an exponential upper bound on the type-two error probability.
For $r\in\{1,\,\infty\}$, let $P\in\mathcal P_n$ be such that $\|p-p_0\|_r\geq M_r \epsilon_{n,r}$.
For
\begin{itemize}
\item $r=1$, when $T_{n,1}\leq M_{0,1}\epsilon_{n,1}$,
$$
\|p-p_0\|_1 \leq \|p-P^n\hat p_n\|_1 + \|\hat p_n-P^n\hat p_n\|_1 + T_{n,1}
 \leq \|p-P^n\hat p_n\|_1 + \|\hat p_n-P^n\hat p_n\|_1  +  M_{0,1}\epsilon_{n,1},$$
\item $r=\infty$, when $T_{n,\infty}\leq M_{0,\infty}\epsilon_{n,\infty}$,
\[\begin{split}
\forall\, x\in\mathcal X,\quad |p(x)-p_0(x)| & \leq \|p-P^n\hat p_n\|_\infty + \|\hat p_n-P^n\hat p_n\|_1 + T_{n,\infty}
\leq \|p-P^n\hat p_n\|_\infty + \|\hat p_n-P^n\hat p_n\|_1 + M_{0,\infty}\epsilon_{n,\infty},
\end{split}\]
which implies that $\|p-p_0\|_\infty \leq \|p-P^n\hat p_n\|_\infty + \|\hat p_n-P^n\hat p_n\|_1 + M_{0,\infty}\epsilon_{n,\infty}$.
\end{itemize}
Summarizing, for $r\in\{1,\,\infty\}$, when $T_{n,r}\leq M_{0,r}\epsilon_{n,r}$, we have
$\|p-p_0\|_r \leq \|p-P^n\hat p_n\|_r + \|\hat p_n-P^n\hat p_n\|_1 + M_{0,r}\epsilon_{n,r}$.
If $\sup_{P\in\mathcal P_n}\|p-P^n\hat p_n\|_r=
\sup_{P\in\mathcal P_n}\|p-K_{J_n}(p)\|_r\leq C_K\epsilon_{n,r}$, we have
$\|\hat p_n-P^n\hat p_n\|_1  \geq \|p-p_0\|_r-\|p-P^n\hat p_n\|_r - M_{0,r}\epsilon_{n,r}
 \geq [M_r- (C_K+M_{0,r})]\epsilon_{n,r}$. Using, as before, McDiarmind's inequality with $P$ playing the same
role as $P_0$, we get that for a constant $\alpha\in(0,\,1)$ small enough and $[M_r-(C_K+M_{0,r})]>0$,
\[\begin{split}\sup_{P\in\mathcal P_n:\,\|p-p_0\|_r\geq M_r \epsilon_{n,r}}
P^n(1-\phi_{n,r})=P^n(\|\hat p_n-p_0\|_r\leq M_{0,r}\epsilon_{n,r})&=P\bigl(\|\hat p_n-P^n\hat p_n\|_1 \geq [M_r-(C_K+M_{0,r})]\epsilon_{n,r}\bigr)\\
&\leq 2 \exp{(-\alpha^2[M_r-(C_K+M_{0,r})]^2n\epsilon_{n,r}^2/\|\Phi\|_1^2)}.
\end{split}\]
We need that $\alpha^2[M_r-(C_K+M_{0,r})]^2/\|\Phi\|_1^2\geq (C+4)$, which implies that
$[M_r-(C_K+M_{0,r})]\geq \alpha^{-1}\|\Phi\|_1\sqrt{C+4}$. This concludes the proof of the first assertion.

If the convergence in \eqref{eq:postconc} holds for $r=1$ and $r=\infty$, then the last assertion of the statement
follows from the interpolation inequality: for every $1< s < \infty$,
$\|p-p_0\|_s\leq \max\{\|p-p_0\|_1,\,\|p-p_0\|_\infty\}$.
\end{proof}

\section{Auxiliary results for Proposition \ref{prop2}}\label{App}

Following \citet{Parzen}, \citet{WatsonLeadbetter}, we adopt the subsequent definition.

\begin{definition}
The Fourier transform or characteristic function of a Lebesgue
probability density function $p$ on $\mathbb{R}$,
denoted by $\tilde p$, is said to \emph{decrease algebraically of degree $\beta>0$} if
\begin{equation*}\label{eq:algebraic}
\lim_{|t|\rightarrow\infty}|t|^\beta|\tilde p(t)|=B_p, \quad 0<B_p<\infty.
\end{equation*}
\end{definition}

The following lemma is essentially contained in the first theorem of section 3B
in \citet{WatsonLeadbetter}.

\begin{lemma}\label{lemma1}
Let $p\in \mathbb{L}^2(\mathbb{R})$ be a probability density with characteristic function
that decreases algebraically of degree $\beta>1/2$. Let $h\in\mathbb{L}^1(\mathbb{R})$
have Fourier transform $\tilde h$ satisfying

\begin{equation}\label{eq:integrability}
I_\beta[\tilde h]:=\int\frac{|1-\tilde h(t)|^2}{|t|^{2\beta}}\di t<\infty.
\end{equation}
Then, $\delta^{-2(\beta-1/2)}\|p-p\ast h_\delta \|_2^2\rightarrow (2\pi)^{-1}B^2_p\times
I_\beta[\tilde h]$ as $\delta\rightarrow0$.
\end{lemma}

\begin{proof}
%\smartqed
Since $p\in \mathbb{L}^1(\mathbb{R})\cap\mathbb{L}^2(\mathbb{R})$, %we have $\hat p\in \mathbb{L}^2(\mathbb{R})$.
then $\|p\ast h_\delta\|_q\leq \|p\|_q\|h_\delta\|_1<\infty$, for $q=1,\,2$. Thus,
$p\ast h_\delta\in \mathbb{L}^1(\mathbb{R})\cap\mathbb{L}^2(\mathbb{R})$.
It follows that $(p-p\ast h_\delta)\in \mathbb{L}^1(\mathbb{R})\cap\mathbb{L}^2(\mathbb{R})$.
Hence,
\begin{eqnarray*}
\|p-p\ast h_\delta\|_2^2=
\frac{\delta^{2\beta-1}}{2\pi}\bigg\{B_p^2 \times I_\beta[\tilde h]
+
\int \frac{|1-\tilde h(z)|^2}{|z|^{2\beta}}[|z/\delta|^{2\beta}|\tilde p(z/\delta)|^2-B_p^2]\,\di z
\bigg\},
\end{eqnarray*}
where the second integral tends to $0$ by the dominated convergence theorem
because of assumption (\ref{eq:integrability}).
\end{proof}
In the next remark, which is essentially due to \citet{Davis77}, section 3, we consider a sufficient condition for a function
$h\in L^1(\mathbb{R})$ to satisfy requirement (\ref{eq:integrability}).
\begin{Rem}
\emph{If $h\in L^1(\mathbb{R})$, then $\int_1^{\infty}t^{-2\beta}|1-\tilde h(t)|^2\,\di t<\infty$ for $\beta>1/2$.
Suppose further that there exists an integer $r\geq 2$ such that $\int x^m h(x)\,\di x=0$, for $m=1,\,\ldots,\, r-1$,
and $\int x^r h(x)\,\di x\neq 0$.
Then,
\[\begin{split}
\frac{[1-\tilde h(t)]}{t^{r}} = - t^{-r} \int \bigg[e^{itx}-\sum_{j=0}^{r-1}\frac{(itx)^j}{j!}h(x)\bigg]\, \di x
= - \frac{i^r}{(r-1)!}\int x^r h(x)\int_0^1(1-u)^{r-1}e^{itxu}\,\di u\,\di x
\rightarrow
-\frac{i^r}{r!}\int x^r h(x)\,\di x,
\end{split}\]
as $t\rightarrow0$. For $r\geq\beta$, the integral $\int_0^1t^{-2\beta}|1-\tilde h(t)|^2\,\di t<\infty$. Conversely, for $r<\beta$, the integral diverges.
Therefore, for $1/2<\beta\leq 2$, any symmetric probability density $h$ with finite second moment is such that
$I_\beta[\tilde h]<\infty$ and condition (\ref{eq:integrability}) is verified.}
\end{Rem}

%%%%%%%%%%%%%%%%%%%%%%%%%%%%%%%%%%%%%%%%%%%%%%%%%%%%%%%%%%%%%%%%%%%%%%%%%%%%%%%%%%%%%%%%%%%%%%%%%%%%%%%%%%%%%%%%%%%%%%%%%%%%%%%%%%%%%%%%

%% References


%%\bibliographystyle{elsarticle-num}
%%\bibliography{<your-bib-database>}



%\bigskip
%\vspace*{1cm}

\nocite{*}
\bibliographystyle{plain}

%\noindent{\bf References}

\end{document}